\documentclass[12pt]{amsart}
\usepackage[hang,small,bf]{caption}

\usepackage{amsmath, amssymb, amsthm, amsfonts, mathrsfs}
\usepackage{graphicx, subfig, float}
\usepackage{cite}

\usepackage{setspace}
\usepackage{url, fullpage}

\newtheorem{prop}{Proposition}

\newtheorem{theorem}[prop]{Theorem}

\theoremstyle{definition}

\begin{document}
\title{A Population-centric Approach to the Beauty Contest Game}
\author{Marc Harper}
\address{University of California Los Angeles}
\email{marc.harper@gmail.com} 

\date{\today}


\begin{abstract}
An population-centric analysis for a version of the $p$-beauty contest game is given for the two-player, finite population, and infinite population cases. Winning strategies are characterized in terms of iterative thinking relative to the population. To win the game one needs to iterate more times than the ambient population, but not too many more times depending on the population size and the value of $p$.
\end{abstract}

\maketitle

\section{Introduction}

The beauty contest game concept originated with John Maynard Keynes \cite{Keynes36} and has been studied in many contexts and populations \cite{nagelsurvey} \cite{bosch2002one} \cite{Grosskopf08} \cite{Sbriglia08}. In a beauty contest game, all players guess a number within a given interval, with the goal of guessing $p$ times the average of all other guesses, where $p$ is a number in the interval $(0,1)$ and guesses are constrained between 0 and 100. For instance, for $p=1/2$ a player attempts to guess what will be half of the average of all guesses (including the player's guess). The game is commonly played with values of $p=1/2$ and $p=2/3$. It will be shown below that a strategy exists to always win for $0 < p \leq 1/2$ in a particular version of the game. This situation is more complex, however, for $p > 1/2$, but the possible winning strategies are characterized in these cases as well.

In \cite{Camerer03}, a cognitive hierarchy model of players of the beauty prize game is given where each player is modeled by how many iterations of reducing the guess to match the average that a player goes through mentally. For example, a 1-player guesses $p$ times the center of the interval (the average score). A 2-player considers that other players may guess $p$ times half of the interval, and so guesses $p$ times the center. Call a player a $k$-player if they guess $p^k c$, where $c$ is the center of the interval. The parameter $k$ need not be an integer. Camerer et al experimented with several groups of individuals to determine the average depth of iteration of a particular population. Interestingly, computer scientists and game theorists showed the most ``steps of thinking'', 3.8 and 3.7 respectively. Players from a version of the game appearing in a newspaper averaged 3.0 steps and high school students 1.6 steps. Refer to the table in \cite{Camerer03} for other groups and variations.

While a model that had the proportion of $k+1$-players decreasing relative to the proportion of $k$-players was used in \cite{Camerer03}, this article takes the point of view common in evolutionary game theory \cite{Hofbauer03} \cite{Gintis}. Given a population of $k$-players, can a $l$-player dominate or invade the population? We characterize successful strategies for all $p$ in multiple cases and give computational results in some more difficult but interesting cases. The analysis relies only on a basic knowledge of single-variable calculus.

In the first section, pairwise interactions in a selective population are considered. In the next section, a different approach is used that captures the fact that the players in this game act simultaneously in both infinite and finite populations. It is assumed that $0 < p < 1$ so as to avoid the degenerate cases $p=0$ and $p=1$.

\section{Head-to-Head Gameplay}

An evolutionary model such as the replicator equation pits players against each other in pairwise interactions, so first consider the two-player interaction of a $k$-player and an $l$-player. Assume that $l = k + m$, i.e. that the $l$-player uses $m$ more steps than the $k$-player, without loss of generality.

How does a $k$-player fare versus a $k+m$ player, for $m > 0$? The target of the guess is
\[ p \cdot \left(\text{average guess}\right) = p \cdot \frac{p^k c + p^{k+m} c}{2} = \frac{c}{2}\left( p^{k+1} + p^{k+m+1}\right) \]
Since $p \in (0,1)$, the following inequalities are valid
\[ c p^{k+m} \leq c p^{k+m+1} \leq \frac{c}{2}\left( p^{k+1} + p^{k+m+1}\right) \leq c p^{k+1} \leq c p^k, \]
and so the distance from the mean for the $k$-player is
\[ c p^k - \frac{c}{2}\left( p^{k+1} + p^{k+m+1}\right) = c p^k \left( 1 - \frac{p}{2} - \frac{p^{m+1}}{2}\right)\]
and for the $k+m$ player is
\[ \frac{c}{2}\left( p^{k+1} + p^{k+m+1}\right) - cp^{k+m}  = c p^k \left(\frac{p}{2} + \frac{p^{m+1}}{2} - p^{m}\right) . \]
To determine which player is closer, take the difference and notice that
\begin{align*}
c p^k & \left[ \left(1 - \frac{p}{2} - \frac{p^{m+1}}{2} \right) - \left(\frac{p}{2} + \frac{p^{m+1}}{2} - p^{m}\right) \right]\\
& = c p^k \left(1 - p - p^{m+1} + p^m \right)\\
&= c p^k (1-p)(1+p^m) > 0. 
\end{align*}

The player that uses more steps $l = k+m$ wins the head-to-head matchup, so the $l$-player dominates. A population of two-player interactions of these two strategies tends toward fixation of the $l$-player since it bests the $k$-player and ties itself, hence has selective advantage in a population of $k$-players and $l$-player interacting pairwise.

\section{Population Games}

The beauty prize game is intended for simultaneous action of all players, not pairwise interactions between randomly-matched players. Still, the goal is to determine if a $(k+m)$-player can succeed in a population of $k$ players, and for what values of $m$ this is true.  Suppose a $(k+m)$-player enters a population of $N$ $k$-players. The target of the guess of the $(k+m)$-player in this case is
\[ p \cdot \left(\text{average guess}\right) = \frac{p}{N+1} \left(N p^k c + p^{k+m} c\right) = c p^{k+1}\left(\frac{N+p^m}{N+1} \right). \]

\subsection{Infinitely Large Population}

Let us first consider the infinitely large population case. As $N \to \infty$, the target guess converges to $c p^{k+1}$ (for any value of $m$ and $p$). Hence a $(k+m)$-player wins the game for any $m \in (0,1]$, since the guess is clearly closer to the target as the are both above the average.  For general $m$, again consider the difference of the distances of the two player types to the target of the guess, and notice that the value of $c>0$ only scales the difference and can be discarded. The difference is
\[ \phi(p) = 1 - 2p + p^m. \]

Intuitively, if $p < 1/2$, the computation of $p$ times the average guess favors lower guessers. Indeed, if $p \leq 1/2$, $\phi(p) > 0$ for all $m$ since $\phi(p) = 1 - 2p + p^m \geq p^m > 0$. In other words, if $p \leq 1/2$, using more steps than the players in the population always wins. Assume now that $p > 1/2$.

It is easy to see that for $m=2$, $\phi(p) > 0$ for all $p$ by factorization. Observe that $\phi$ is decreasing in $m$, since the derivative with respect to $m$ is negative for $1/2 < p < 1$, so for $0 < m \leq 2$, $\phi(p) > 0$. For values of $m > 2$, it is possible for $\phi$ to be negative.

\begin{figure}[H]
    \centering
    \includegraphics[scale=0.5]{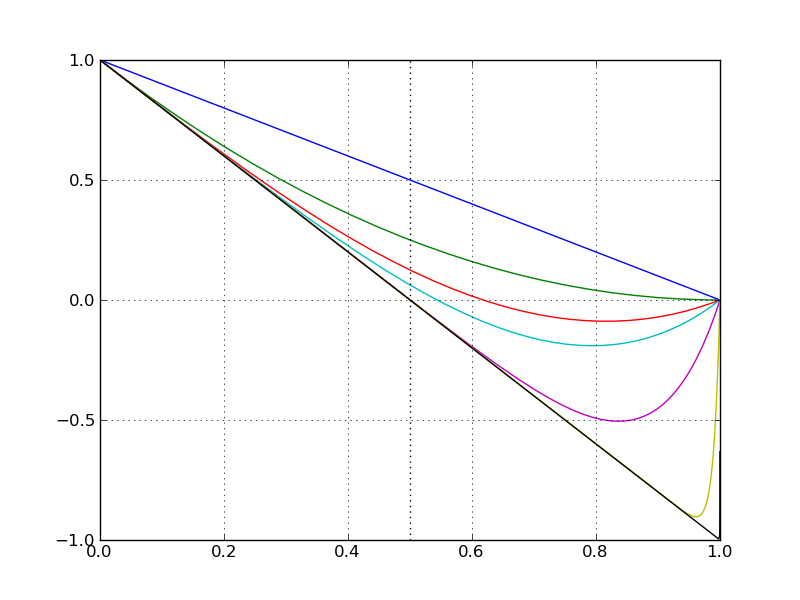}
    \caption{Plot of $\phi(p)$ for $m = 1, 2, 3, 4, 10, 100, \infty$ (from top to bottom). For $m>2$, these curves have a root in the interval $(1/2, 1)$, indicating a change in advantage between the players.}
    \label{phi_plot}
\end{figure}

For integer values of $m > 2$ and $1/2 < p < 1$, consider the factorization
\[ \phi(p) = (1-p) \left( 1 - p - p^2 - \cdots - p^{m-1} \right). \]
There is a root at the degenerate case $p=1$. By Descartes' rule of alternating signs, the rightmost quantity has at most one positive root. If this root lies in the interval $(1/2, 1)$ then there is a change in which player type is favored for a given $m$. This is true for all $m > 2$ as shown by the following argument.

\begin{prop}
For $m > 2$, $\phi(p)$ has a unique root in the interval $(1/2, 1)$. 
\label{inf_pop}
\end{prop}
\begin{proof}
For constant $m$, $\phi$ has a local minimum at $p_{min} = \left( \frac{2}{m} \right) ^{ \frac{1}{m-1}}$, which is in the interval $(0,1)$ if and only if $m>2$ (by inspection).  For $m > 2$, this minimum is negative, which follows from algebraic manipulation. Since $\phi(p) > 0$ for $p < 1/2$, and the minimum $\phi(p_{min})$ is negative, by the intermediate value theorem there is a value $p^{*}$, with $1/2 < p^{*} < p_{min}$, for which $\phi(p^{*}) = 0$. This value is unique since there are no other local extrema in the interval and $\phi(1) = 0$, and so no other changes in the sign of $\phi$ since $\phi$ is continuous.
\end{proof}

Notice that $p^{*} \to 1/2$ as $m \to \infty$ from the form of $\phi$, and so the interval $(0, p^{*})$ in which a $(k+m)$-player can win shrinks to $(0, 1/2]$ as $m$ gets larger. See Figure \ref{p_star_p_min}.

%
%

\begin{figure}[H]
    \centering
    \includegraphics[scale=0.5]{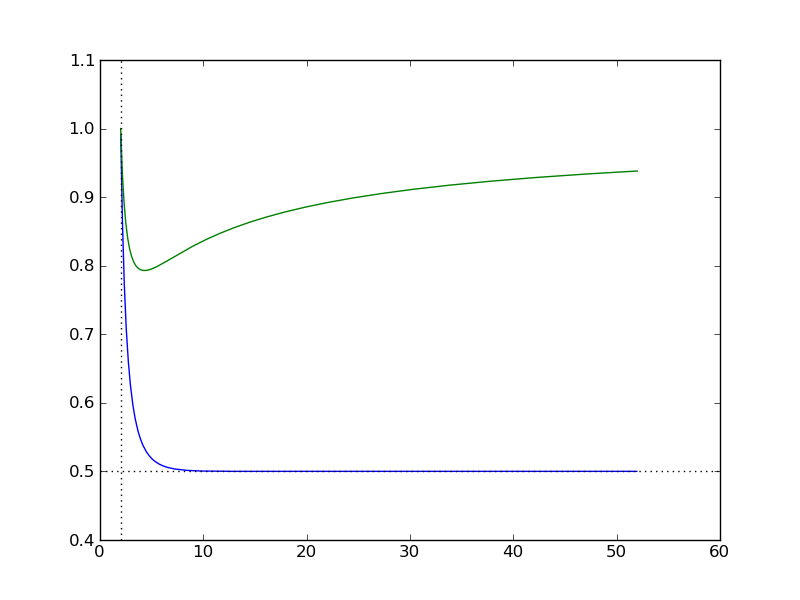}
    \caption{Plot of $p_{min}(m), p^{*}(m)$. Notice that $1/2 < p^{*}(m) < p_{min}(m) < 1$ for all $m > 2$.}
    \label{p_star_p_min}
\end{figure}

In a large population, the key to winning the game is to use within 2 steps of the opponent strategy for all $p$. It is possible to win with larger values of $m$, that is by using more than 2 additional steps, but this depends heavily on the value of $p$ and the range of admissible $p$ drops off quickly as $m$ increases, as is evident from Figures \ref{phi_plot} and \ref{p_star_p_min}. Hence winning the game hinges on an accurate estimation of the number of steps used in the opponents strategy and using within approximately two steps more in forming a guess. It pays to be more clever than your opponents, but not too much more clever.

\subsection{Finite Populations}

Let us drop the assumption of an infinitely large population and derive population-size dependent results. For populations of size $N$ with $2 \leq N < \infty$, a derivation similar to the above gives the function $\phi(p)$ as
\[ \phi(p) = 1 - 2p \frac{N+p^m}{N+1} + p^m \]

For large $m$, the function reduces to
\[ \phi(p) = 1 - 2\frac{N}{N+1} p. \] In this case, $\phi(p) > 0$ if $p < 1/2 + 1/(2N)$ and less than zero if $p > 1/2 + 1/(2N)$ since $\phi$ is linear in $p$ for fixed $N$. Similarly, for large $N$, $\phi$ converges to the infinite-population function as expected.

\begin{figure}[H]
    \centering
    \includegraphics[scale=0.5]{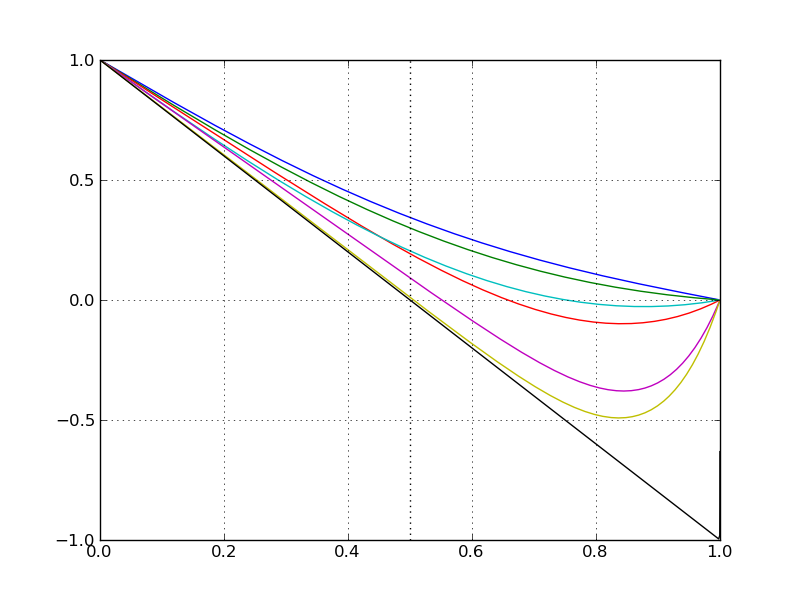}
    \caption{Plot of $\phi(p)$ for $(N, m) = (3,3), (4,3), (5,5), (10, 3), (10, 10), (100, 10), (10000,10000)$ (from top to bottom, note that $(4,3)$ and $(5,5)$ intersect). For small values of $N$ and $m$ these curves may not have a root in the interval $(1/2, 1)$. This differs from the infinitely large population (in which there is a root if $m > 2$).}
    \label{phi_N_m_plot}
\end{figure}

In the general case, notice that for very small values of $N$ and $m$, e.g. $N=3$ and $m=3$ that the function $\phi(p)$ is always positive. It takes either a larger value of $m$ or a larger value of $N$ for $\phi$ to have a root, but for essentially the same reasons as Proposition \ref{inf_pop}, when the root exists it is unique. Again, since $0 < p < 1$ and $m > 0$,
\[ \phi(p) = 1 - 2p \frac{N+p^m}{N+1} + p^m > 1 - 2p + p^m, \]
and so $\phi(p) > 0$ if $p \leq 1/2$ as in the infinite population case. Though it is difficult to solve for the minimum value explicitly so as to be able to use the intermediate value theorem as in the infinite population case, it is follows that for large $N$ or large $m$ that there exists a root $p^{*} > 1/2$ because of the large $m$ and large $N$ limits and since $\phi$ is continuous.

To investigate small values of the parameters, for each $N$ let $m^{*}(N)$ denote the smallest value of $m$ such that there is a root, given that there is at least one such $m$. Similarly, for each $m$ let $N^{*}(m)$ be the smallest value of $N$ such that there is a root, given that there is at least one such $N$. There exists a solution for all $m > m^{*}(N)$ and for all $N > N^{*}(m)$ if these values exist, respectively. Figure \ref{n_star_m_star} demonstrates that very small values of $N$ or $m$ require large values of the other, but once the $N$ exceeds approximately 10, there are solutions for $m \geq 3$.

\begin{figure}[H]
    \centering
    \includegraphics[scale=0.5]{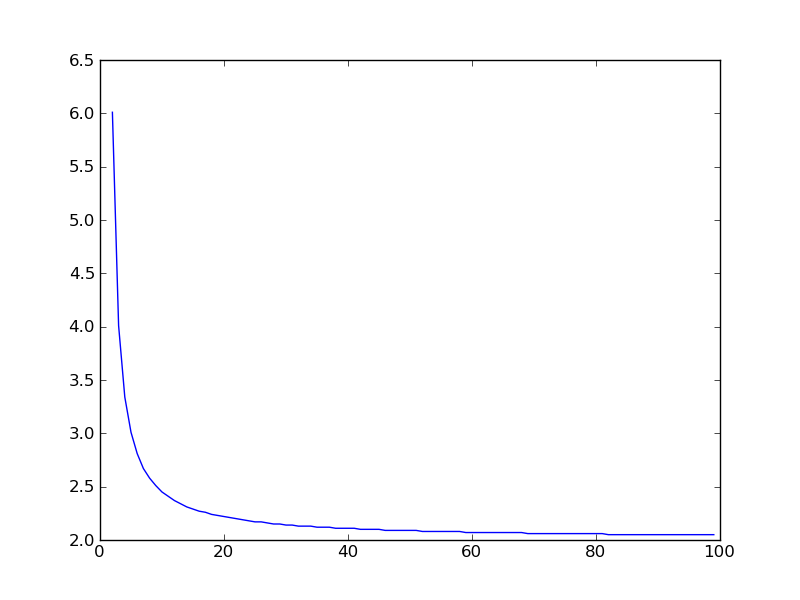}
    \caption{Plot of $m^{*}(N)$, critical values of $m$ for which a solution $p^{*}$ exists. For all pairs $(N, m)$ in the region above the curve, $\phi$ has a root in the interval $(1/2, 1)$. Below the curve, no such solution exists.}
    \label{n_star_m_star}
\end{figure}

\begin{prop}
Both $N^{*}(m')$ and $m^{*}(N')$ exist for all $N' > 1$ and $m' > 2$. $m^{*}(N') = m'$ if and only if $N^{*}(m') = N'$.
\end{prop}
\begin{proof}
To show existence, consider the limits given above. Since $\phi$ is continuous and decreasing in $m$ and $n$ for any particular $p$, it follows that for any fixed $N' > 1$ or $m' > 2$ that there exists a root for large enough $m$ or $N$, respectively. Moreover, both $N^{*}$ and $m^{*}$ are decreasing, since a larger value of $N$ potentially permits a smaller value of $m$, and vice versa, by inspection of $\phi$.

Suppose that $m', N'$ is a pair such that $\phi$ has a solution, and let $m^{*}(N') = m'$. Then it follows immediately that $N^{*}(m') \leq N'$. To obtain equality, suppose for the sake of contradiction that $N^{*}(m') < N'$. Then by the monotonicity of $N^{*}$, there is a value $m'' < m'$ such that $m''$ and $N'$ have a solution which contradicts the minimality of $m'$. (In other words, $N'$ is above the curve of $N^{*}$, so there are points to the left along the line $N=N'$ still above the curve.) The other implication is completely analogous.
\end{proof}

Combining the discussion and previous proposition yields the following theorem.

\begin{theorem}
Let $(N, m)$ lie in the region above the curve $N^{*}$ (or equivalently $m^{*}$). Then $\phi(p)$ has a solution in the interval $(1/2, 1)$.
\end{theorem}

For a finite population, the key conclusion remains basically the same -- to win the beauty prize game, a player must be more clever than its opponents, but not too much more clever, especially if the group of opponents is large. The smaller the number of opponents, the more clever a player can be and still win. In other words, the baseline of the population and its size limits how successful clever strategies (in the sense of more iterations) can be.

Finally, can $(k+m)$-players invade a population of $k$-players? If the population consists of $M$ $(k+m)$-players and $N-M$ $k$-players, the target guess is
\[ c p^{k+1} \left(\frac{N+M(p^m - 1)}{N}\right).\]
This quantity is decreasing in $M$ since $(p^m - 1) < 0$, so the population dynamic favors the invaders as they push the winning guess further from the ambient $k$-players. Hence if the values of $N$, $p$, and $m$ are such that a single mutant strategy can invade, it will fixate under a selective or imitative dynamic, and is susceptible to still more clever mutants. If mutants of $l$-players of arbitrarily increasing $l$ continue to invade, the population will converge to the Nash equilibrium of the game, $l=\infty$, pushing the average guess down to zero. This is very similar to phenomena of schoolyard one-upmanship of guessing successively larger quantities, to which a clever child claims ``infinity!'' (often met with the reply of ``infinity plus one!'').

\section{Discussion}

An extended analysis could investigate games of mixed populations of $k$-players for different values of $k$, in the replicator and simultaneous play approaches. The preceding results show, however, that this analysis would need to be extremely meticulous because the winner of a pairwise or population-wide contest between $k$ and $l$ players depends heavily on the difference $k - l$, and the value of $p$. In particular, for many values of $p$, a difference of less than two favors the lower guesser, but a larger difference may not.

Simulations suggest that in a population of more than two player types, under selective action (e.g. a discrete replicator dynamic), fixation tends to occur for the player with the most steps, or if there is a large gap between the number of steps in the highest step players, the player with the second most number of steps, depending on initial conditions. Simulations were run with fitness determinations given by either an all-or-nothing payout or a payout proportional to the inverse of the distance of the player-type's guess to the target guess.

\bibliography{ref}

\begin{thebibliography}{1}

\bibitem{bosch2002one}
Antoni Bosch-Domenech, Jose~G Montalvo, Rosemarie Nagel, and Albert Satorra.
\newblock One, two,(three), infinity,...: Newspaper and lab beauty-contest
  experiments.
\newblock {\em American Economic Review}, pages 1687--1701, 2002.

\bibitem{Camerer03}
Colin Camerer, Teck Ho, and Kuan Chong.
\newblock Models of thinking, learning, and teaching in games.
\newblock {\em Experimental Economics}, 93(2):192--195, 2003.

\bibitem{Gintis}
H.~Gintis.
\newblock {\em Game Theory Evolving}.
\newblock Princeton University Press, 2000.

\bibitem{Grosskopf08}
B.~Grosskopf and R.~Nagel.
\newblock The two-person beauty contest.
\newblock {\em Games and Economic Behavior}, (62):93--99, 2008.

\bibitem{Hofbauer03}
Josef Hofbauer and Karl Sigmund.
\newblock Evolutionary game dynamics.
\newblock {\em Bulletin of the American Mathematical Society}, 40(4), 2003.

\bibitem{Keynes36}
John~Maynard Keynes.
\newblock {\em The General Theory of Employment, Interest and Money.}
\newblock Harcourt Brace and Co., 1936.

\bibitem{nagelsurvey}
R~Nagel.
\newblock A survey on experimental ‘beauty contest games’: Bounded
  rationality and learning, budescu d., erev i., zwick r., games and human
  behavior: Essays in honor of amnon rapoport, 1998.

\bibitem{Sbriglia08}
P.~Sbriglia.
\newblock Revealing the depth of reasoning in p-beauty contest games.
\newblock {\em Experimental Economics}, 11(2), 2008.

\end{thebibliography}
\bibliographystyle{plain}

\end{document}